\documentclass[11pt]{article}
\usepackage[utf8]{inputenc}
\usepackage{authblk}
\usepackage{stackrel}
\usepackage[T1]{fontenc}
\usepackage{todonotes}

\usepackage{color}
\definecolor{darkgreen}{rgb}{0,0.5,0}
\definecolor{darkblue}{rgb}{0,0,0.8}
\definecolor{darkred}{rgb}{0.8,0,0}

\usepackage{hyperref}
\hypersetup{
   unicode=false,          
   colorlinks=true,        
   linkcolor=darkred,          
   citecolor=darkblue,        
   filecolor=magenta,      
   urlcolor=black           
}

\title{\textbf{Comparison Dynamics in Population Protocols}}
\author{ Dan Alistarh}
\affil[1]{IST Austria}
 \author[1]{Martin Töpfer}
 \author[2]{Przemysław Uznański}
 \affil[2]{Institute of Computer Science, University of Wrocław, Poland}
\date{}

\usepackage{graphicx, amsmath, amsfonts, comment, amsthm, ulem, enumitem}
\usepackage{fullpage}

\DeclareMathOperator*{\E}{\mathbb{E}}
\newcommand{\Var}{\mathrm{Var}}

\newtheorem{theorem}{Theorem}
\newtheorem{lemma}{Lemma}

\newtheorem{definition}{Definition}
\newtheorem{observation}{Observation}
\newtheorem{corollary}{Corollary}

\renewcommand{\emph}[1]{\textit{#1}}

\newcommand{\alg}{\texttt{PopComp}}

\renewcommand{\paragraph}[1]{\vspace{.2em}\noindent\textbf{#1}}

\newcommand{\bigo}{\mathcal{O}}

\begin{document}
\maketitle

\begin{abstract}
    There has recently been a surge of interest in the computational and complexity properties of the \emph{population model}, which assumes $n$  anonymous, computationally-bounded nodes, interacting at random, with the goal of jointly computing global predicates. 
Significant work has gone towards investigating \emph{majority} or \emph{consensus} dynamics in this model: that is, assuming that every node is initially in one of two states $X$ or $Y$, determine which state had higher initial count. 
    In this paper, we consider a natural generalization of majority/consensus, which we call \emph{comparison}: in its simplest formulation, we are given two \emph{baseline} states, $X_0$ and $Y_0$, present in any initial configuration in fixed, but possibly small counts. One of these states has higher count than the other: we will assume $|X_0| \ge C |Y_0|$ for some constant $C > 1$.
    The challenge is to design a protocol by which nodes can quickly and reliably decide on which of the baseline states $X_0$ and $Y_0$ has higher initial count. 

    We begin by analyzing a simple and general dynamics solving the above comparison problem, which uses $\bigo( \log n  )$ states per node, and converges in $\bigo( \log n )$ (parallel) time, with high probability, to a state where the whole population votes on opinions $X$ or $Y$ at rates proportional to the initial  concentrations of $|X_0|$ vs. $|Y_0|$. 
    We then describe how this procedure can be bootstrapped to solve comparison, i.e. have every node in the population reach the ``correct'' decision, with probability $1 - o(1)$, at the cost of $\bigo(\log \log n)$ additional states.
    Further, we prove that this dynamics is \emph{self-stabilizing}, in the sense that it converges to the correct decision from \emph{arbitrary initial states}, and \emph{leak-robust}, in the sense that it can withstand spurious faulty reactions, which are known to occur in practical implementations of population protocols. 
    Our analysis is based on a new martingale concentration result relating the discrete-time evolution of a population protocol to its expected (steady-state) analysis, which should be a useful tool when analyzing opinion dynamics and epidemic dissemination in the population model.

\end{abstract}

\thispagestyle{empty}
\section{Introduction}
Population protocols are a model of distributed computation in which a set of $n$ simple agents, modeled as identical state machines, cooperate to jointly compute predicates over the system's initial state. Agents have no control over their interaction pattern: they interact in pairs, chosen by an external scheduler. A common assumption, which we also adopt, is that the interaction schedule is uniform random across all possible node pairs.
Since its introduction~\cite{AADFP06}, the population model has become a  popular way of modeling distributed computation in various settings, from animal populations, to wireless networks, and chemical reaction networks. 
Significant attention has been given to the computational power of population protocols~\cite{AAER07, CMNPS11}, as well as determining the complexity thresholds for fundamental problems, such as leader election and majority. For recent surveys of the area, please see~\cite{elsaesser2018recent, AG18}, as well as the more recent~\cite{becchetti2020consensus}, which focuses specifically on consensus dynamics in the population model 

One classic example of the algorithmic power of population protocols is the elegant \emph{three-state approximate majority algorithm}. Discovered independently by~\cite{AAE08, PVV09}, this simple dynamics has been implemented in synthetic DNA~\cite{CD13}, and has been linked to the fundamental cell cycle biological process~\cite{CCN12}. 
The problem it addresses is \emph{majority}, where it is assumed that all agents are initially in one of the states $A$ or $B$, and agents must converge on a consensus decision as to which one of the two had higher initial count. 
This is done via the following interactions:
    $$A + B \rightarrow C + C, \,\,
    A + C \rightarrow A + A, \textnormal{   and   } B + C \rightarrow B + B.$$ 
Intuitively, if both ``strong'' opinions ($A$ or $B$) interact, then they both move to the ``undecided'' state $C$, while either of the ``strong'' opinions $A$ or $B$ turns an undecided $C$ agent to its side. 
Angluin et al.~\cite{AAE08} showed that this simple algorithm  has surprisingly strong properties: it converges to the correct majority decision with high probability (w.h.p.), as long as the initial difference between the initial states is $\Omega( \sqrt{n \log n} )$, in time that is poly-logarithmic in $n$, and that it can even withstand limited \emph{Byzantine} failures. ({In the following, we adopt the standard definition of high probability to be at least $\geq 1 - 1 / n^c$, where $c \geq 1$ is a constant.})

A related problem arising in practical scenarios is \emph{robust detection}~\cite{alistarh2017robust}, in which nodes aim to determine if a distinct \emph{detectable} state $D$ is present in or absent from the population. 
Moreover, the problem requires that the algorithm be \emph{robust to leaks}~\cite{Leaks}, roughly defined as low-probability faulty reactions in which any state implemented by the algorithm may appear spuriously. ({Leaks are meant to model the impact of the laws of chemistry on the algorithm execution, which might for instance reverse reactions with some small probability. 
We detail the definition of leaks and their impact on the execution in the model section.}) 
\cite{alistarh2017robust} shows a robust detection protocol which satisfies both these requirements. \cite{DBLP:conf/stoc/DudekK18} considers the same detection problem, showing that any self-stabilizing protocol for detection requires $\Omega(\log \log n)$ states per node if the algorithm is to converge in poly-logarithmic time, and $\omega(1)$ states if the goal is sublinear time. Second, they show that detection can in fact be solved in  $\bigo(1)$ states, by a protocol which \emph{does not stabilize}, as some states may keep oscillating between very small and large counts.


\paragraph{The Robust Comparison Problem.} 
In this paper, we consider a generalization of both  \emph{majority} and \emph{robust detection}, which we call  \emph{robust comparison}. For simplicity, we first describe a simple \emph{static} version of the task: we are given two \emph{baseline} states, $X_0$ and $Y_0$, present in any initial configuration in fixed, possibly small, e.g. logarithmic, counts.  Importantly, one of these states has higher count than the other: we assume that $|X_0| \ge C |Y_0|$ for some constant $C > 1$. 
The goal is to design a protocol which allows nodes to quickly decide on which of these baseline states has higher count. More precisely, we will require nodes to output the correct answer with a prescribed probability $>1/2$. 
\newline
In the \emph{dynamic} variant of the problem, we ask that the algorithm should be \emph{self-stabilizing}, in the sense that it should converge to the correct decision \emph{from arbitrary initial states}. In particular, in the dynamic version, we allow the counts of the baseline states $X_0$ and $Y_0$ to be changed adversarially during the execution, as long as these  counts remain stable for a sufficiently long period allowing the algorithm to stabilize on an output. We will show that polylogarithmic parallel time is in fact sufficient. In addition, we will also show a variant of the protocol that is \emph{leak-robust}. 

\paragraph{Problem Motivation.} The comparison problem has not been considered at this level of generality before. Detection~\cite{alistarh2017robust, DBLP:conf/stoc/DudekK18} is obviously a  special case of comparison, where one of the baseline states has \emph{zero} count. 
Moreover, the classic approximate majority problem, e.g.~\cite{AAE08}, can be seen as a \emph{static, one-shot} special instance of comparison, in which both baseline states have initial count $\Theta(n)$, and we wish to determine which of the two had higher initial count. Of note,  comparison in the general case when the baseline counts may be $o(n)$ is \emph{strictly harder} than \emph{approximate majority}: constant-state, polylog-stabilization time algorithms exist for approximate majority~\cite{AAE08}, but, following the lower bound of~\cite{DBLP:conf/stoc/DudekK18}, no such algorithm may exist for comparison. 

%

\paragraph{Contribution.} We describe and analyze simple and general dynamics for solving robust comparison in population protocols, and provide concentration bounds on its convergence using a new analysis technique. 

The basic algorithm, which we call \alg{}, uses $\bigo( \log n )$ states per node, stabilizes to the correct answer in parallel time $\bigo( \log n )$ from \emph{any} initial configuration, and is robust to leaks. 
For simplicity, we describe the algorithm in the \emph{static} case below, and will generalize the presentation in the later sections. 
Assume some given set of agents in baseline states $X_0$ and $Y_0$, whose relative counts the algorithm needs to compare. For now, we will assume that these states are \emph{immutable}, i.e. the algorithm's interaction rules do not change their counts,  since their relationship is what we need to determine. 
Without loss of generality, assume $|X_0| > |Y_0|$.
The algorithm will implement sequences of ``detector'' states $X_1, X_2, \ldots, X_s$ and $Y_1, Y_2, \ldots, Y_s$, where $s = \log n + \Theta(\log \log n)$ is a parameter, as well as a neutral state $N$. 

The intuitive role of the indexed \emph{strong} $X_i$ and $Y_i$ states is to measure how long the interaction chain is between the current agent and an $X_0$ or $Y_0$ node at any given point.
For example, any node which interacts directly with $X_0$ will move to state $X_1$, and symmetrically, any node which interacts directly with $Y_0$ will move to state $Y_1$. The key interaction is between a node in state $X_j$ or $Y_j$, which interacts with a node $X_i$ of \emph{lower} index $i < j$. 
In this case, the former agent will be part of a \emph{shorter} interaction chain with respect to $X_0$, moving to state $X_{i + 1}$, while the latter agent increases the length of its chain by one, moving to $X_{i + 1}$ as well. 
Generalizing, we obtain a series of reactions of the type:
\begin{eqnarray}
\label{eqn:increase}
\forall_{s > j > i > 1}\quad  &\ \ X_i + X_j  \rightarrow X_{i + 1} + X_{i + 1}, \\
&\ Y_i + Y_j  \rightarrow Y_{i + 1} + Y_{i + 1},  \notag\\
&\ \ X_i + Y_j  \rightarrow X_{i + 1} + X_{i + 1},  \notag\\
&Y_i + X_j  \rightarrow Y_{i + 1} + Y_{i + 1}. \notag
\end{eqnarray}

Notice that $\bigo(\log n)$ is a natural upper bound for the length of an interaction chain, since every agent is  $\bigo(\log n)$ ``hops'' away from  $X_0$ or $Y_0$, with high probability. The key observation is that we can reliably use the \emph{relative sizes} of these interaction chains to distinguish between the baseline states: agents are more likely to be ``closer'' to the more populous state, rather than to the competitor. 
We leverage this as follows. Let us cap the maximum level at $s = \log n + \Theta(\log \log n)$. Nodes continue to increase their level or reset it to a previous one, according to Equation~\ref{eqn:increase}, as long as the level's value is  $\leq s$. As soon as the length of the chain would increase past $s$, agents move to the \emph{neutral} state $N$, at which point they stop influencing other agents in terms of their choice.  A neutral agent can become non-neutral only if it interacts with another $X_i$ or $Y_i$ agent with $i < s$, and it resets the length of its chain to $\leq s$.
%
%

\paragraph{Analysis.} As is often the case in population protocols, this algorithm is intuitive; however, its recursive structure requires careful analysis. 
A natural first approach would be a ``steady-state'' analysis, in which one writes out the expected counts of agents of every type and the relationships between them, assuming stable counts. One then solves this system of constraints in order to determine the expected counts at ``equilibrium.'' 
However, at best, this approach yields \emph{expected} bounds on the state counts, and cannot characterize the concentration of state counts at some given point in the execution. 
In particular, in the case of our algorithm, since consecutive level counts are highly correlated, characterizing their concentration is challenging---if not impossible---using this approach. Concretely, notice that even minor fluctuations of the interaction counts for the lower-level states---e.g., the less populous baseline states $Y_0$ happen to interact more frequently than $X_0$ states for a brief interval---can lead to ``inversion cascades'' at later stages in the chain. 
More generally, linking steady-state behavior with exact algorithm dynamics is known to be difficult in population protocols, and even for some basic algorithms only steady-state behaviour is known~\cite{DV12, alistarh2017robust}. 
 
We introduce a new approach to circumvent this limitation, based on two technical ideas. 
The first is that, even though the state counts at various levels are correlated, their evolution roughly has super-martingale-type behavior, oscillating around its average value, with variability (``noise'') due to the natural variance of state counts at previous levels. 
(See Section~\ref{sec:warmup} for a detailed walk-through.) 
A tempting approach then is to apply a Bernstein-type martingale concentration inequality~\cite{Bernstein} to the level counts in order to characterize their concentration  around their expectation. 
However, known versions of this concentration result do not apply to our setting, in particular due to the presence of noise. 

We overcome this problem by proving a new  concentration bound, which could be useful more generally. 
This result allows us to bound the influence of variability at previous levels onto the counts at a certain level $\ell$, and to prove concentration for each of the level counts. 
Iterating, we obtain that, if the base level counts $X_0$ and $Y_0$ are separated by a large enough multiplicative constant $C_1$, 
then the counts at the \emph{last level} will also be separated by a multiplicative constant $C_2$, w.h.p.
This result allows us to show \emph{fast} convergence: level counts will recover to concentrate close to their expected mean in poly-logarithmic parallel time. 
In turn, this result opens up several extensions.  

\paragraph{Extensions.}
The first extension of the above dynamics boosts the probability that an agent identifies the correct output state from the \emph{constant} one postulated above, to $1 - o(1)$. 
Thus, all but a sub-constant fraction of the agents reach the correct answer.  
Boosting is achieved via a general approximate counting mechanism, which has each agent use $\bigo( \log \log n )$ additional state to sample the population and determine the majority state with higher confidence.

As a second  extension, we exhibit a non-trivial space-time trade-off for variants of this protocol. 
For instance, we exhibit two protocol variants which employ $o(\log n)$ and $\bigo( \log \log n ) $ states, and ensure convergence in parallel time $\bigo( \log^{\Theta(1) }  n )$ and $n^{o(1)}$, respectively. These protocols show that it is possible to perform comparison in sub-linear time using less than logarithmic states per agent, coming close to the lower bound of~\cite{DBLP:conf/stoc/DudekK18}. 
The analyses of these variants leverage different instantiations of our concentration theorem. 

Third, we show that our algorithm is \emph{leak-robust}, i.e. can withstand spurious reactions. Again, this property follows by applying the concentration theorem with modified parameters to account for faulty reactions. 


\paragraph{Related Work.}
Our work is part of a wider research effort studying majority, consensus and leader election dynamics in population protocols. 
For algorithms with \emph{exact/deterministic} correctness guarantees, tight or almost-tight space-time trade-offs are now known, thanks to continuous progress over  the past decade, e.g.~\cite{DV12, MNRS17, alistarh2017time, AAG18, BRKKR18, BEFKKR18,podcmajority, berenbrink2020optimal}. In brief, the \emph{logarithmic} space and time complexity thresholds appear to be tight for exact majority, e.g.~\cite{AAG18, BRKKR18, BEFKKR18, podcmajority,DBLP:journals/corr/abs-2106-10201}.   
Constant-state solutions with fast convergence (but no stabilization) are known for both approximate and exact majority~\cite{DBLP:journals/corr/abs-1802-06872}.

A close examination shows that the existing results on the \emph{exact majority} problem, enumerated above, are  not comparable with our results. 
Specifically, these algorithms ensure stronger guarantees, but also require a stronger set of assumptions. 
Superficially, exact majority algorithms may appear stronger since they \emph{deterministically} address the case of a \emph{constant} initial gap, e.g. $|X_0| - |Y_0| = \bigo(1)$. However, these algorithms assume $|X_0| + |Y_0| = n$, whereas we usually focus on lower counts. Further, these algorithms are \emph{not self-stabilizing}. (Specifically, most exact algorithms rely on state invariants which will be broken if we start in arbitrary initial state.) 
To illustrate this point, notice that there is a complexity gap between exact majority algorithms, which are known to require $\Omega( \log n )$ states to stabilize in sublinear time~\cite{AAG18}, and our algorithm, which employs $\bigo(\log \log n)$ states, and requires $n^{o(1)}$ time.  
We believe that it may be possible for some exact majority techniques, e.g.~\cite{AAG18}, to be extended to the setting where $|X_0| + |Y_0| = o(n)$, although this does not appear straightforward, and the resulting algorithms would not be self-stabilizing.  

By contrast to the exact case, the complexity of \emph{approximate} majority---which may converge to the wrong answer with some probability---and that of \emph{dynamic} ones---where the input may change during the execution---are not well-studied. 
For approximate majority, this may be because the classic three-state approximate majority protocol~\cite{AAE08} already unifies several desirable properties: fast convergence, robustness to Byzantine faults, and an optimal state space size. 

Subsequent work by d'Amore, Clementi, and Natale~\cite{d2020phase} considered a variant of approximate consensus known as opinion dynamics in a model of probabilistically noisy interactions, and established that a phase transition occurs for specific values of the error probability. Although related, the problem and model are different, and therefore the results are incomparable.    

In this paper, we generalize the approximate majority problem to the case where the two initial states have small initial counts, and the goal of the other agents is to determine which baseline state/signal is more populous/stronger. We further generalize to the case where the counts of these initial states may change. 
The references technically closest to ours are the recent work on detection dynamics~\cite{alistarh2017robust, DBLP:conf/stoc/DudekK18}, which we have covered in the previous section. (We cover the relation between our results and the detection lower bound of~\cite{DBLP:conf/stoc/DudekK18} in Section~\ref{sec:discussion}.) 
The algorithm we analyze is a generalization of the detection dynamics considered by~\cite{alistarh2017robust}: 
in particular, if we merged the  $X$ and $Y$ states, we would obtain a similar algorithm to the basic version of \alg{}. 

We make several contributions relative to the latter reference. 
First, we consider a more general problem, which is closer to consensus dynamics than to detection/rumor-spreading. 
Second, we provide a significantly more accurate, and technically challenging analysis. 
Specifically,~\cite{alistarh2017robust} only provides an expected-value analysis for the detection problem. 
In contrast, we are able to provide strong concentration bounds for comparison, which can be further boosted via additional mechanisms, 
and provide a thorough exploration of time-space trade-offs for this problem.   
In addition, our analysis introduces a powerful and novel generalized Bernstein-type inequality, which should be a useful addition to the analysis toolbox of population dynamics.

\begin{figure*}[t]
    \centering
    \hspace{1cm}
    \begin{minipage}[t]{6cm}
        \centering
        \includegraphics[scale=0.45]{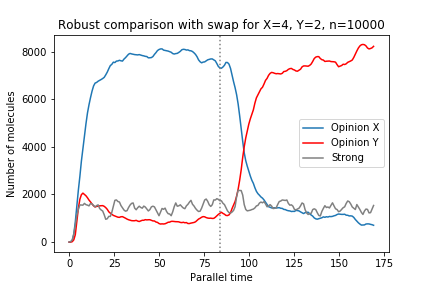}
    \end{minipage}
    \hspace{1em} \qquad
    \begin{minipage}[t]{6cm}
        \centering
        \includegraphics[scale=0.45]{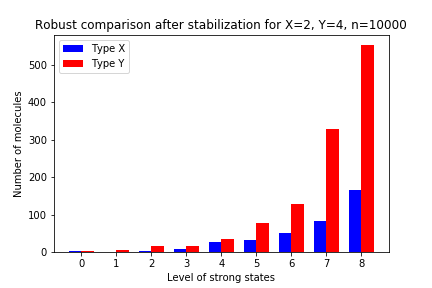}
    \end{minipage}
    \caption{Implementation results. In the left figure, we depict the counts of molecules with the $X$ opinion (blue) versus the $Y$ opinion (red) from the starting from an initial state where $X_0 > Y_0$. At parallel time $80$ (dotted), we switch these numbers, and record the change in counts. The gray line (bottom) counts the number of agents in \emph{strong} states. The right figure considers the same setup (after the switch), but counts the number of agents in each level of the strong states after stabilization.}
    \label{fig:algo}
\end{figure*}

\section{System Model and Problem Statement}
\label{sec:model}

\paragraph{Population Protocols.} 
A population protocol is a distributed system with $n \geq 2$ nodes, also called \emph{molecules} or \emph{agents}. 
Nodes execute a deterministic state machine with states from a finite set $S_n$,
whose size may be a function of $n$. 
Nodes are anonymous, so agents in the same state are identical and interchangeable. 
Consequently, the state of the system at any point is characterized by the number of nodes in each state with non-zero count. 
Formally, a \emph{configuration} $c$ is a function 
  $c: S_n \to \mathbb{N}$, where $c(s)$ represents the 
  \emph{number of agents in state $s$}.
Nodes interact in pairs, according to an outside entity called the \emph{scheduler}.
In this paper, we will assume a \emph{uniform random} scheduler, which picks every possible interaction pair uniformly at random, which corresponds to having a well-mixed solution.   

An algorithm, also known as a \emph{population protocol}, is defined as follows. 
We define the set $I_n$ of all allowed initial configurations of the protocol 
  for $n$ agents, a finite set of output symbols $O$, a transition function 
  $\delta_n : S_n \times S_n \rightarrow S_n \times S_n$, 
  and an output function $\gamma_n : S_n \rightarrow O$.
The system starts in one of the initial configurations $i_n \in I_n$ (clearly, $|i_n| = n$),
  and each agent keeps updating its local state following interactions with other agents, 
  according to the transition function $\delta_n$.
The execution proceeds in \emph{steps}, where in each step a new pair of agents 
  is selected uniformly at random from the set of all pairs. 
Each of the two agents updates its state according to the function $\delta_n$. 

\paragraph{Time, Space, Stabilization, and Self-Stabilization.} 
Our basic notion of \emph{steps} counts the number of interactions until some given  predicate holds on the entire population. Parallel time is defined as total number of pairwise interactions divided by the number of nodes $n$. 
We measure \emph{space} as the number of states which can be implemented by each node. 
We say that a population protocol is \emph{self-stabilizing}~\cite{AAFJ08} if it is guaranteed to converge to a set of output configurations which satisfy a given predicate from \emph{any} initial configuration, and for which every extension also satisfies the given predicate. The parallel time to reach those output configurations is the \emph{stabilization time}. 

\paragraph{Leaks and Robustness.} 
We now recall the definition of \emph{leak reactions (leaks)}, following~\cite{alistarh2017robust}. 
Given the above, any population protocol can be specified as  a sequence of transition rules of the form
$$ X + Y \rightarrow Z + T.$$

Given the set of such transitions defining a protocol,~\cite{alistarh2017robust} partitions protocol states into  \emph{catalytic states}, which never change count following \emph{any} reaction: for instance, state $C$ is catalytic if it only participates in reactions of the type $X + C \rightarrow Y + C,$ where $X$ and $Y$ are arbitrary. 
By contrast, \emph{non-catalytic states} can change their count, for instance to be created or transformed by the protocol into other states. 
In a nutshell, \emph{leaks} are spurious reactions which can consume and create arbitrary \emph{non-catalytic} species, from other \emph{non-catalytic} species.
Leaks are induced by the basic laws of chemistry. For instance, by the law of reversibility, every interaction has some (low) probability of being reversed; by the law of catalysis, every catalytic reaction can also occur in the absence of the catalyst state. In practice, leaks can cause any molecule type implemented by the algorithm to appear spuriously during its execution, with some low probability.

More formally, a leak is a reaction of the type $S \rightarrow S',$ where $S$ and $S'$ denote arbitrary \emph{non-catalytic} states.
For generality, in the following we will assume that the exact leak reactions are chosen \emph{adversarially}, but that their rate, that is, their probability of occurring at a given moment, will be upper bounded by a fixed parameter $\gamma$. 
An algorithm which maintains its correctness guarantees in spite of leaks is called \emph{leak-robust}~\cite{alistarh2017robust}. Notice that protocols such as the four-state exact majority algorithm~\cite{DV12} are \emph{not} leak-robust, since the correctness of their output crucially depends on having exact molecule counts throughout the execution. 

\paragraph{The Comparison Problem.} 
We are given two \emph{baseline} states, $X_0$ and $Y_0$, present in any initial configuration in possibly small counts. 
One of these states has higher count than the other: w.l.o.g., we will assume that $|X_0| \ge C |Y_0|$ for some constant $C > 1$. 
For simplicity, we first discuss the \emph{static} version of the problem, in which the initial counts of  $X_0$ and $Y_0$ are fixed, and do not change throughout the execution. 

The goal of the (static) comparison problem is to design a protocol which can decide on which of the states $X_0, Y_0$ has higher count. Thus, each node should output $X$ if $|X_0| \ge C |Y_0|$, and $Y$ if $|Y_0| \ge C |X_0|$. Otherwise, any output X/Y is allowed. 
We will consider protocols with probabilistic outputs, i.e. a node might return the correct output with probability $2/3$ or $1 - o(1)$. The \emph{convergence time} is the time it takes for the protocol to reach a state where the output predicate holds with a given probability.  

In the \emph{self-stabilizing} variant of the problem, we assume that the protocol can be initialized in an arbitrary state, and should still converge to the correct decision. In particular, we will consider the interesting \emph{dynamic input} case where the proportion of $X_0 / Y_0$ changes during the execution. 
The \emph{self-stabilization time} is the time it takes for the protocol to converge to the correct decision from such an arbitrary configuration. 
In the \emph{robust} variant of the problem, we consider static inputs (i.e. fixed $X_0$ and $Y_0$ counts) but allow spurious leak reactions, as described above.

\section{The \alg{} Robust Comparison Algorithm}

\subsection{The Baseline Algorithm}

In this section, we present the baseline variant of the algorithm, which ensures a \emph{constant} separation between the two states, in favor of the more numerous one, with high probability. 
In the next section, we will build on this algorithm to boost the fraction of nodes which correctly identify the majority to $(1 - o(1))$. 

\paragraph{Algorithm Description.} 
 Initially, we have a fixed set of nodes in states $X_0$ and $Y_0$, whereas all other nodes may have arbitrary states (for simplicity, we may assume that all are in a neutral state $N$). 
Thus, each node's state will be either $X_0, X_1,\dots,X_s$, $N$ or $Y_0, Y_1,\dots,Y_s$, where $s$ is a \emph{level parameter}, whose value is specified later in the analysis. States $X_0,\dots,X_s$ will correspond to output $X>Y$ with decreasing ``confidence" (symmetrically for $Y_i$ states) and $N$ is a neutral state (it roughly corresponds to both states $X_{s+1}$ and $Y_{s+1}$ being merged). We call a molecule \emph{strong} if its state is not $N$. The state changes according to the following rules:

\begin{equation*}
\begin{aligned}
\text{For all}\ &1\le i \le s: \\
X_0 + X_i &\rightarrow X_0 + X_1 \\
X_0 + Y_i &\rightarrow X_0 + X_1 \\
Y_0 + X_i &\rightarrow Y_0 + Y_1 \\
Y_0 + Y_i &\rightarrow Y_0 + Y_1 \\
\end{aligned}
\quad\quad
\begin{aligned}
\text{For all}\ &1\le i < s: \\
X_i + N &\rightarrow X_{i+1} + X_{i+1} \\
Y_i + N &\rightarrow Y_{i+1} + Y_{i+1} \\
X_i + Y_i &\rightarrow X_{i+1} + Y_{i+1}\\
\text{~}\\
\end{aligned}
\end{equation*}

\begin{equation*}
\begin{aligned}
\text{For all}\ &1\le i \le j \le s; i\neq s:\\
X_i + X_j &\rightarrow X_{i+1} + X_{i+1} \\
Y_i + Y_j &\rightarrow Y_{i+1} + Y_{i+1} \\
\end{aligned}
\quad\quad
\begin{aligned}
\text{For all}\ &1\le i < j \le s:\\
X_i + Y_j &\rightarrow X_{i+1} + X_{i+1} \\
Y_i + X_j &\rightarrow Y_{i+1} + Y_{i+1} \\
\end{aligned}
\end{equation*}

\begin{equation*}
\begin{aligned}
X_s + X_s \rightarrow& N + N \\
X_s + Y_s \rightarrow& N + N \\
Y_s + Y_s \rightarrow& N + N \\
Y_s + X_s \rightarrow& N + N \\
\end{aligned}
\quad\quad
\begin{aligned}
X_s + N &\rightarrow N + N \\
Y_s + N &\rightarrow N + N \\
X_0 + N &\rightarrow X_0 + X_1\\
Y_0 + N &\rightarrow Y_0 + Y_1\\
\end{aligned}
\end{equation*}

The intuition is that the state of molecules is used to spread the information about the number of initial molecules in $X_0$ and $Y_0$ states, which never change, among all other molecules, while we maintain approximately the ratio $\frac{|X_i|}{|Y_i|} \approx \frac{|X_0|}{|Y_0|}$. This is done by confidence levels $X_1,\dots,X_s$ (resp. $Y_1,\dots,Y_s$). A molecule decreases its confidence by one during each reaction but it spreads its information to the less confident molecule in the reaction. When the confidence passes the threshold $s$, the molecule moves to a neutral state $N$. We will show that the number of molecules in consecutive levels roughly doubles at every level, with high probability.
We present an experimental illustration in Figure~\ref{fig:algo}.

\paragraph{Guarantees.} 
Our main technical result is a general analysis of this dynamics, which can be summarized as follows.

\begin{theorem}
\label{th:mainresult-brief}
For any choice of constants $C_2>C_3>1$, there exists constant $C_1$ such that if $|X_0|, |Y_0| \geq C_1 \log n$, and $\frac{|X_0|}{|Y_0|} \ge C_2$, the  algorithm stabilizes in parallel time $\mathcal{O}(\log n)$ to a configuration where $\frac{\sum_{i=1}^s |X_i|}{\sum_{i = 1}^s |Y_i|} \ge C_3$ with high probability. 
\end{theorem}

Intuitively, the results says that, given any multiplicative gap $C_2 > 1$ between the baseline states, a target gap $C_3$ between the total state counts can be achieved as long as the baseline counts are a large enough constant multiple $C_1$ of $\log n$. 
We next discuss this result and its extensions in detail. 

\subsection{Algorithm Extensions and Discussion}
\label{sec:discussion}

\paragraph{Trade-off between initial and total counts.}
From the point of view of the gap between baseline states, the guarantees in Theorem \ref{th:mainresult-brief} are governed by a trade-off which is similar in nature to Chernoff bounds. 
We stated the result for the lowest possible values of the initial counts $|X_0|$ and $|Y_0|$ (up to constants). Our result leads to ``constant'' boosting of the gap, which is the best possible in this general case. 
However, we can strengthen the result to allow for lower base separation. 
For instance, if we allow the baseline states to have slightly higher counts $|X_0|,|Y_0| = \Omega(\log^2 n)$, then it is enough to set $C_2,C_3 = 1+o(1)$  where $C_2>C_3$ with $o(1)$ as small as $O(1/\sqrt{\log n})$.

It is reasonable to ask whether the count separation between the baseline states may be \emph{constant additive} instead of  multiplicative. Our analysis can be applied to the case where the gap is additive by simply taking the parameter $C_2$ to be a function of $|Y_0|$. However, as noted when discussing the relationship to exact majority algorithms, the analysis will break down in the case when the gap is an additive constant \emph{and} the counts $|X_0|$ and $|Y_0|$ are logarithmic. This is simply because random fluctuations in the counts at intermediate levels can negate the relative advantage of the more populous baseline state, leading to fluctuations in the output decision. 

\paragraph{Boosting Precision.} The algorithm described in the previous subsection ensures constant separation--roughly, we can guarantee with a proper choice of parameters that, say, at least $2n/3$ of all molecules have the correct output, and at most $n/3$ have the wrong output.
Now we describe a way of amplifying this correctness guarantee. 
We describe it with respect to our algorithm, but the transformation is generic and would apply to any comparison algorithm.

Assume that, in addition to their state, molecules are equipped with a counter that contains an integer value in the interval $[-m, m]$, where $m$ is a parameter. 
The counter is increased by one if a molecule reacts with a strong molecule of type $X_i$, and decreased by one if it reacts with a molecule of type $Y_i$. 
If a molecule reacts with a molecule in state $N$, the counter remains unchanged. 
The output function $\gamma_n$ maps all states with a positive counter to output $X>Y$ and all states with a negative counter to $Y>X$. 

Note that, when the confidence levels stabilize in the baseline algorithm, the counter should function similarly to a random walk biased towards the majority. 
More precisely, it is biased towards $+\log \log n$ if $|X_0| > |Y_0|$, and vice versa. 
Because there are $\Omega\left(n\right)$ strong molecules, each one reacts with enough strong molecules, and therefore the random walk should quickly converge to its stationary distribution. 
The stationary distribution will give us the estimate that there are only $\bigo(n/\log n)$ molecules with wrong value of counter in expectation.

There are two ways of implementing the above dynamics. 
The first method has every molecule participate in the counting process. 
This requires increasing the number of states to $\bigo(s \cdot m)$, but has the advantage that each molecule is participating in the output.
Second, similarly to~\cite{GP16}, one can split the population initially into two roughly equal-size parts.  
The first half implements the original amplification algorithm, while the second half consists of molecules implementing the random-walk counter. 
Thus, the number of states becomes $\bigo(s + m)$, but with the disadvantage that a constant fraction of all molecules do not produce any output at all.

The above construction ensures that the algorithm stabilizes in time $\bigo(\log n)$ to a configuration where at most $\bigo(n / \log n)$ have the incorrect output. It uses $\bigo(\log n \log \log n)$ states. The proof is provided in Section~\ref{sec:precision}.

\paragraph{Time vs. space tradeoff.}
In Section~\ref{sec:trade-off}, we explore different variants of this algorithm which trade off a lower state space for higher convergence time. 
Interestingly, we will show that there exist a variant with $o(\log n)$ states per node, which converges in $\textnormal{ polylog } n$ time, and a variant with $\bigo( \log \log n )$ states per node which still converges in sub-linear time. 
Since these variants require a more careful re-definition of the protocol, we present them separately in the corresponding section. \cite{DBLP:conf/stoc/DudekK18} showed the following  time-space complexity trade-off for detection/comparison:
\begin{theorem}[\cite{DBLP:conf/stoc/DudekK18}, corollary of Theorem 4.1]
Any protocol that solves detection using $T$ parallel time by convergence to a stationary distribution\footnote{More formally, by this the authors mean that the algorithm converges to states from which the state counts cannot vary by more that $n^{\alpha}$, where $\alpha > 0$ is a constant. Our algorithm satisfies this condition. } requires $\Omega\left(\left(\log \frac{\log n}{\log T}\right)^c\right)$ states\footnote{This bound follows from the extended statement of this theorem, mentioned at the end of its relevant section, where non-constant number of states lowerbounds are considered.} for some absolute positive constant $c$.
\end{theorem}

Our protocol gives a trade-off for comparison (and thus detection) with $\Theta\left(\frac{\log n}{\log \frac{T}{\log n}} + \log \log n\right)$ states for any $T \ge \log n$, leaving an exponential gap between lower and upper bounds.

\section{Analysis of the Baseline Algorithm}

\label{sec:conc}

In this section we will focus on the concentration properties of $|X_i|$ and $|Y_i|$, the number of molecules of type $X_i$ and $Y_i$, respectively, for each level. 
Intuitively, given initial counts of $X_0, Y_0$, the argument establishes (i)  upper- and lower-bounds on the counts of $X_i,Y_i$ in the ``steady state'' of the protocol, 
(ii) shows that the protocol concentrates around those bounds, and (iii) that concentration occurs quickly. 

\paragraph{Notation.} Denote 
$R_i = \bigcup_{j \le i} (X_j + Y_j)$. Also denote $x_i=\frac{|X_i|}{n}$, $y_i=\frac{|Y_i|}{n}$ and 
$r_i = \frac{|R_i|}{n} = \sum_{j \le i} (x_j + y_j)$.

To specify value of some variable after precise number of interactions, we add $(t)$ after the variable -- i.e. $x_i(t)$ denotes the probability that a randomly chosen molecule $t$ steps of protocol is of type $X_i$. 

%
%

\subsection{The Detection Problem}
\label{sec:warmup}

The goal of this section is to develop some of the intuition behind the analysis, as well as some preliminary results, by providing bounds on the joint count at each level, denoted by $U_i$. 
Recall that $s \geq 1$ is the maximum level achievable by states $X_i$ and $Y_i$. 
We begin with the observation that if we replace all states $X_i$ and $Y_i$ with $U_i$ in the algorithm, then the interaction rules become: 

\begin{equation*}
\begin{aligned}
\text{For all}\ &1\le i \le s: \\
U_0 + U_i &\rightarrow U_0 + U_1 \\
\text{~}\\
\text{For all}\ &1\le i < s: \\
U_i + N &\rightarrow U_{i+1} + U_{i+1} \\
U_s + N &\rightarrow N + N \\
\end{aligned}
\quad\quad
\begin{aligned}
\text{For all}\ &1\le i \le j \le s; i\neq s:\\
U_i + U_j &\rightarrow U_{i+1} + U_{i+1} \\
U_s + U_s &\rightarrow N + N \\
\end{aligned}
\end{equation*}
Note that this closely matches the detection dynamics of~\cite{alistarh2017robust}: intuitively, in this case, we are not trying to compare the counts of two species, but instead trying to detect the presence of a single species $X_0 + Y_0$ in the initial solution. 
We note that the analysis in~\cite{alistarh2017robust} only provides \emph{expected} bounds on the species count at every level. 
Thus, the preliminary results of this section illustrate our analysis technique by tightening the bounds for this detection algorithm to characterize concentration. 
In turn, concentration bounds are essential to analyze the behavior of the comparison dynamics we consider. 

Recall that $r_i = |R_i| / n$ for any level $i$. We begin by introducing some auxiliary variables $\tilde{r}_i$, for each level $i$, which are intuitively the steady-state (expected) values to which the level counts should converge in the limit. 
Let also $\tilde{R}_i = n \cdot \tilde{r}_i$. 
Note that defining the values $\tilde{r}_i$ directly in terms of the convergence of the process can be difficult, so instead we will directly provide an operational definition for them. 
More precisely, we define these values recursively as follows: 
$$\tilde{r}_0 = r_0 = \frac{|X_0|+|Y_0|}{n} \textnormal{ and } \tilde{r}_{i+1} = \tilde{r}_i \cdot (2 - \tilde{r}_i), \forall s > i \geq 0,$$
where the recurrence follows from the observation that an agent is in state $R_{i+1}$ iff in its last interaction, at least one of the interacting agents was in state $R_i$.\footnote{One can also write ODE for concentration of level counts, and observe that \emph{conditioned} on constant value of $r_i$, $r_{i+1}$ converges exponentially towards $r_i \cdot (2-r_i)$.}
We can expand this recursion to obtain the following estimates for these level counts. 

\begin{observation}
For any $i \geq 1$, it holds that $\tilde{r}_i = 1 - \left(1 - \tilde{r}_0\right)^{2^i}.$
In particular, we have $\tilde{r}_i = \Theta(\min(1, 2^i \cdot \tilde{r}_0)) .$
\end{observation}

Our goal will be to provide a \emph{concentration bound} for the values of the level counts $r_i$ to match these steady-state values. 
Broadly, our setup is as follows. We will fix a level index $c \geq 0$ and time $t$, such that, at this time, the level counts $R_i$ at levels $i \leq c$ are well-concentrated around their means $\tilde{R}_i$, with high probability. 
Then, we will show that there exists a time $t' \geq t$, such that, with high probability, the level count at level $c + 1$ is concentrated around its own predicted mean $\tilde{R}_{c + 1}$. 
 More precisely, let us fix a level $c$ and a time $t$, and assume that there exists a constant $\xi_c > 0$ such that  $\Big||R_c(t)| - \tilde{R}_c\Big| \le \xi_c \tilde{R}_c$, with high probability.  
We will proceed to prove that  there exists a constant $\xi_{c + 1}$ and a time $t'\geq t$  such that $\Big||R_{c + 1}(t')| - \tilde{R}_{c + 1}\Big| \le \xi_{c+1} \tilde{R}_{c + 1}$ given a sufficiently large time interval $T = t' - t$. 

The argument will begin by analyzing the evolution of the level counts at time $t+1$.
In particular, denote by $\Delta_1(t),\Delta_2(t),\Delta_3(t),\Delta_4(t) \in \{0,1\}$ the indicator variables for the following events at step $t$, which govern the evolution of $R_{c + 1}(t + 1)$:
\begin{itemize}
\item $\Delta_1(t) = 1$ iff first reacting agent was from $R_{c+1}$,
\item $\Delta_2(t) = 1$ iff second reacting agent was from $R_{c+1}$,
\item $\Delta_3(t) = 1$ and $\Delta_4(t) = 1$ iff any of the reacting agents were from $R_{c}$ (in fact $\Delta_3(t) = \Delta_4(t)$ is always satisfied).
\end{itemize}

We obtain the following recurrence on the expected value of $R_{c + 1}(t + 1)$: 
\begin{align*}
\mathbb{E}\Big[|R_{c+1}(t+1)|\ |\ t\Big] &= \mathbb{E}\Big[|R_{c+1}(t)| - \Delta_1(t) - \Delta_2(t) + \Delta_3(t) + \Delta_4(t)\ |\  t\Big]\\
&= |R_{c+1}(t)| - 2 r_{c+1}(t) + 2 [1 - (1-r_{c}(t))^2] \\
&= \left(1-\frac{2}{n}\right) |R_{c+1}(t)| + \frac{2}{n} A(t)
\end{align*}
where we define $A(t) = n \cdot [1 - (1-r_{c}(t))^2]$.
Second, we bound the variance by direct calculation:
\begin{align*}
\mathrm{Var}\Big[|R_{c+1}(t+1)|\ |\  t\Big] &= \mathrm{Var}\Big[|R_{c+1}(t+1)| - |R_{c+1}(t)|\ |\ t\Big] \\ 
 & = \mathrm{Var}[\Delta_3(t)+\Delta_4(t)-\Delta_1(t)-\Delta_2(t)\ |\  t]\\
&\le 4 (r_{c+1}(t) + r_{c+1}(t) + (1-(1-r_c(t))^2) \\ & + (1-(1-r_c(t))^2)) = 8 \frac{|R_{c+1}(t)|}{n} + 8 \frac{A(t)}{n}.
\end{align*}

Finally, we use the induction hypothesis to bound the deviation of $\tilde{R}_{c + 1}$ from $A(t)$, with high probability, as 
\begin{align*}
|A(t)-\tilde{R}_{c+1}| &=\Big| n [1 - (1-r_{c}(t))^2] - n  (1 - (1-\tilde{r}_c)^2) \Big| \\ & \le n  \xi_c \tilde{r}_c (2 - \tilde{r}_c) \leq 2 \xi_c \tilde{R}_{c+1}.
\end{align*}

\subsection{The Concentration Theorem} 

We now take a step back, and examine the claims we have already proven, and their relationship to our target. 
We wish to obtain a concentration bound on the level count $|R_{c + 1}|$ in terms of its predicted steady-state value $\tilde{R}_{c + 1}$. 
We have a handle on the expected value of $|R_{c + 1}|$ and on its variance, but these values critically depend on the quantity $A(t)$. 
At the same time, we also have a strong probabilistic bound on how much $A(t)$ can vary, by the last inequality. 
A natural candidate to establish a concentration bound on $|R_{c + 1}|$ would be to recognize that it has super-martingale behavior, and apply a Bernstein-type inequality for its concentration around its mean. 
However, it is hard to see how to apply this result to our setting, in particular due to the presence of the ``noise'' term $A(t)$. 
Fortunately, we are able to prove the following   concentration result instead. 

\newcommand{\maintoolA}{
Fix parameters $n \geq 1$ and $a \leq n$ with $a = \Omega( \log n )$, and $\varepsilon \leq 1$. Further, fix constants $\lambda, \gamma,\delta,\eta = \bigo(1)$. 
Let $t \geq t_0$ denote time, and let $A(t),B(t) \in [0,n]$ be stochastic processes such that for all time steps $t \ge t_0$ the following hold:}
\newcommand{\maintoolB}{
\begin{enumerate}
\item $|A(t) - a| \le \varepsilon a$,
\item $\E[ B(t+1)\ |\ A(t),B(t) ] = (1-\frac{\lambda}{n})  B(t) + \frac{\lambda}{n}  A(t)$,
\item $|B(t+1)-B(t)| \le \gamma$,
\item $\Var[ B(t+1)\ |\ A(t),B(t) ] \le \delta  \frac{B(t)}{n} + \eta  \frac{A(t)}{n}$.
\end{enumerate}}
\newcommand{\maintoolC}{
Then there exists an interval length $T' = \Theta(\frac{1}{\lambda} n \log n \log \log n)$ such that for any $t' \ge t_0 + T'$ the following holds with high probability:
$$|B(t') - a| \le  \varepsilon a + \bigo(c_1 \sqrt{a \log n} + \log n),$$
for $c_1 = \sqrt{\frac{\delta+\eta}{\lambda}}$.}
\begin{theorem}
\label{th:general_martingale}
\maintoolA\maintoolB\maintoolC
\end{theorem}

The proof of this result is technical, and is deferred to the Appendix. 
To complete our exposition, notice that this result closely matches our set of previous derivations for $|R_{c + 1}|$, while relation (1) holds w.h.p. for the previous level $c$ as part of the induction step. 
More precisely, we can follow the above derivations and plug in $a = \tilde{R}_{c+1}$, and $\varepsilon = 2\xi_c$, $\lambda = 2$, $\delta = \eta  = 8$, and $\gamma = 1$, to obtain the following concentration result on the level counts after a sufficiently long time has passed. 

\newcommand{\onephaseonelevel}{Fix a level index $c < s$, an initial time $t_0$, and let $T = C  n \log n \log \log n$ where $C$ is large enough constant. 
Fix a constant $\xi_c < 1$ and assume that for any step $t \in [t_0, t_0 + T]$ it holds that 
the level count $R_c$ is always $\xi_c$-concentrated around $\tilde{R}_c$, that is 
 $\Big||R_c(t)| - \tilde{R}_c\Big| \le \xi_c \tilde{R}_c$.
Then there exists a constant $\xi_{c+1} = \xi_c + \bigo(\sqrt{\frac{\log n}{\tilde{R}_{c+1}}})$ such that, for any $t \ge t_0 + T$, with high probability, $\Big||R_{c+1}(t)| - \tilde{R}_{c+1}\Big| \le \xi_{c+1} \tilde{R}_{c+1}$.}
 \begin{lemma}
\label{lem:onephaseonelevel}
\onephaseonelevel
\end{lemma}

Finally, we unroll the recursion for a fixed level $c$, and obtain that the following concentration bound should hold after a given point in time.
Note that level zero is always perfectly concentrated around $\tilde{R}_0$. 

\begin{corollary}
\label{cor:detection-bound}
	Given a level $c \geq 1$ and a fixed initial time $t_0$, there exists an absolute constant $\xi < 1$ and a time interval length $T = \Theta( c n \log n \log \log n )$ and  such that 
	for any $t \geq t_0 + T$, it holds with high probability that $|R_{c + 1}(t)| \in \Big[(1 - \xi)\tilde{R}_{c + 1}, (1 + \xi)\tilde{R}_{c + 1}\Big]$. 
\end{corollary}


\subsection{Step Two: The Comparison Process}
\label{sec:comparison-analysis}

We now proceed to analyze the core of our comparison algorithm. We leave aside the voting amplification component, which we analyze separately in Section~\ref{sec:precision}.  
The strategy is a more complex version of the one from the previous section: we derive bounds on the level counts of states $X_i$ and $Y_i$, for each state in turn. 
We will focus on the derivation for $X_i$, since the case of $Y_i$ is symmetric. 

Let $x_i = |X_i| / n$, and $y_i = |Y_i| / n$, for every level $i$. We begin by defining estimate values $\tilde{x}_i$ to which the level counts should concentrate in the steady-state: 
$$\tilde{x}_0 = x_0 = \frac{|X_0|}{n}\quad\quad\quad\quad\tilde{x}_{i+1} = \tilde{x}_i \cdot (2 - \tilde{r}_i - \tilde{r}_{i-1});$$
$$\tilde{y}_0 = y_0 = \frac{|Y_0|}{n}\quad\quad\quad\quad\tilde{y}_{i+1} = \tilde{y}_i \cdot (2 - \tilde{r}_i - \tilde{r}_{i-1}).$$
These values are computed by following the recursion suggested by steady-state analysis: for an agent to end up in state $X_{i+1}$, it needs to be either in state $X_i$ and be the first reagent in interaction with any of $X_{i},\ldots,X_s, Y_{i},\ldots,Y_s$, or the second reagent in interaction with any of $X_{i+1},\ldots,X_s, Y_{i+1},\ldots,Y_s$.
We unroll the recursion to obtain a well-informed guess as to the values around which these variables should concentrate. 

\begin{observation}
There is $\frac{\tilde{x}_0}{\tilde{x}_0 + \tilde{y}_0} = \frac{\tilde{x}_i}{\tilde{x}_i + \tilde{y}_i}$. It can be verified by induction that $\tilde{x}_i + \tilde{y}_i = \tilde{r}_i - \tilde{r}_{i-1}$.
\end{observation}

The rest of this section will be dedicated to proving the following concentration result on the level counts. 

\begin{lemma}
\label{lem:xyonelevel}
Let $c < s$ be a level index and let $T = Cn \log n  \log \log n$ where $C$ is large enough constant. Assume that during all steps $t \in [T]$ it holds that 
 $\forall_{i \le s} \Big||R_i(t)| - \tilde{R}_i\Big| \le \xi \tilde{R}_i$ for $\xi < 1$ with $\xi$ defined as in Corollary~\ref{cor:detection-bound}, and that $|X_c(t) - \tilde{X}_c| \le \epsilon_c \tilde{X}_c$ for some $\epsilon_c < 1$.
Then, for $t \ge T$, there exists a value $\epsilon_{c+1} = \epsilon_c + 2 \varepsilon \frac{\tilde{r}_c}{1-\tilde{r}_c} + \bigo(\sqrt{\frac{\log n}{\tilde{X}_{c+1}}})$ such that, with high probability, it holds that 
$\Big||X_{c+1}(t)| - \tilde{X}_{c+1}\Big| \le \epsilon_{c+1} \tilde{X}_{c+1}.$
\end{lemma}
\begin{proof}
Fix a level index $c \geq 0$ and time $t$, such that, at this time, the level counts $X_i$ at levels $i \leq c$ are well-concentrated around their means $\tilde{X}_i := n \tilde{x}_i$, with high probability. 
We show that there exists a time $t' \geq t$, such that, with high probability, the level count at level $c + 1$ is concentrated around its own predicted mean $\tilde{X}_{c + 1}$. 
Fix a level $c$ and a time $t$, and assume that there exists a constant $\epsilon_c > 0$ such that  $\Big||X_c(t)| - \tilde{X}_c\Big| \le \epsilon_c \tilde{X}_c$, with high probability.  
We will proceed to prove that  there exists a constant $\epsilon_{c + 1}$ and a time $t'\geq t$  such that $\Big||X_{c + 1}(t')| - \tilde{X}_{c + 1}\Big| \le \epsilon_{c+1} \tilde{X}_{c + 1}$ given a sufficiently large time interval $T = t' - t$. 

The argument will begin by analyzing the evolution of the level counts at time $t+1$.
We define  $\Delta'_1(t),\Delta'_2(t),\Delta'_3(t),\Delta'_4(t) \in \{0,1\}$ as indicator variables for the following events at step $t$:
\begin{itemize}
\item $\Delta'_1(t) = 1$ iff the first reacting agent was from $X_{c+1}$;
\item $\Delta'_2(t) = 1$ iff the second reacting agent was from $X_{c+1}$;
\item $\Delta'_3(t) = 1$ iff the first reacting agent was from $X_{c}$ and second reacting agent had a level $> c$, 
or the first reacting agent had a level $\geq c - 1$ and the second reacting agent was from $X_c$;
\item $\Delta'_4(t) = 1$ iff the first reacting agent had level $> c$ and the second reacting agent is from $X_c$, or if the first reacting agent is from $X_c$, and the second reacting agent has level $\geq c$. 
\end{itemize}

Notice that these events cover all the cases where the count of $X_{c + 1}$ might change in this step. 
As before, the plan is to set up the usage of the Concentration Theorem for the random variable $|X_{c + 1}|$. 
For this, we will characterize its mean and variance at step $t + 1$, assuming that the counts at the previous levels are well-behaved, which we can safely assume by the induction step.
By careful calculation, we obtain:

\begin{align*}
\E\Big[|X_{c+1}(t+1)|\ |\ t\Big] = \left(1-\frac2n\right) |X_{c+1}(t)| + \frac 2n A'(t),
\end{align*}

\noindent where we defined $A'(t) = n \cdot  x_c(t)[2-r_{c}(t)-r_{c-1}(t)]$. Further, we have:
\begin{align*}
\mathrm{Var}\Big[|X_{c+1}(t+1)|\ |\  t\Big] \leq  8 \frac{|X_{c+1}(t)|}{n} + 8 \frac{A'(t)}{n}.
\end{align*}

\noindent Another careful upper bound argument yields that
\begin{align*}
|A'(t) - \tilde{X}_{c+1}| 
\le \left( 2 \varepsilon \frac{\tilde{r}_c}{1-\tilde{r}_c} + \varepsilon_c\right) \tilde{X}_{c+1}.
\end{align*}

\noindent At this point, we have enough data to invoke Theorem~\ref{th:general_martingale}, which guarantees that after $T = \bigo(n \log n \log \log n)$ steps we have
$$\Big||X_{c+1}(T)| - \tilde{X}_{c+1}\Big| \le\left( 2 \epsilon \frac{\tilde{r}_c(t)}{1-\tilde{r}_c(t)} + \epsilon_c\right)  \tilde{X}_{c+1} + \bigo( \sqrt{ \tilde{X}_{c+1} \log n}).$$
\end{proof}

We can then iterate this result to obtain the separation result for the proportion of agents supporting either opinion:

\begin{theorem}
\label{th:mainresult}
Let $T = C n \log^2 n \log \log n$ where $C$ is large enough constant. 
Assume that $|R_0| = \bigo\left( \frac{n}{(\log \log n)^2} \right)$ and $|R_0| = \Omega(\log n)$.  
For appropriately chosen constants $C_1, C_2 > \frac12$, if $|X_0| \ge C_1 (|X_0| + |Y_0|)$, then the total count of the population of agents of opinion ``X'', formally 
$P_X = \sum_{i=1}^s X_i$, will satisfy $P_X > C_2 n$, with high probability. 
\end{theorem}
\begin{proof}
Consider the minimal parameter $d$ such that $\tilde{R}_{d} \ge 0.9 n$. For this value, it will hold that $\tilde{R}_{d} \le 0.99 n$ and $d = \log_2 n + \Theta(1)$. 
By Corollary~\ref{cor:detection-bound}, after a time interval of length  $T_1 = \Theta(n \log^2 n \log \log n)$ all values $|R_i|$ satisfy $|R_i| = (1 \pm \varepsilon) \tilde{R}_i$ for $\varepsilon$ a constant that can be made arbitrarily close to 0 (the cost is traded off against the constant hidden in $|R_0| = \Omega(\log n)$). 
After that time, we repeatedly apply Lemma~\ref{lem:xyonelevel} for the first $d'$ levels of $X_i$. The guarantee for opinions $X$ is that $|X_i| = (1 \pm \varepsilon') \tilde{X}_i$, where
$$\varepsilon' = \sum_{i \le d'} 2\varepsilon \frac{\tilde{r}_i}{1-\tilde{r}_i} + \bigo\left(\sum_{i \le d'} \sqrt{\frac{\log n}{\tilde{X}_i}}\right).$$
We note that, by the geometric sum progression (since only constant number of terms satisfy $0.5 n \le \tilde{r}_i \le 0.9 n$:
$$\sum_{i \le d} 2\varepsilon \frac{\tilde{r}_i}{1-\tilde{r}_i}   \le \sum_{i \le d} 100 \varepsilon \tilde{r}_i = \Theta(\varepsilon) $$
and since $\forall_{i \le d}$ we have that  $\tilde{X}_i = \Omega( \log n)$, the second term is also an arbitrarily small constant, we have that $\varepsilon'$ is also constant that can be arbitrary small. We then observe that $P \ge \sum_{i \le d} |X_i| \ge (1-\varepsilon') \sum_{i \le d} \tilde{X}_i = (1-\varepsilon') C_1 \tilde{R}_i \ge (1-\varepsilon') C_1 0.9 n$, and since $C_1$ can be chosen to be large  and $\varepsilon'$ to be small, this is at least $C_2 n$ for some $C_2 > \frac12$.
\end{proof}

%
\subsection{Bootstrapping convergence time}
\label{sec:bootstr}
We now show how to bootstrap on the results in the previous section, and prove convergence within $\mathcal{O}(n \log n)$ interactions, shaving off the additional logarithmic factors.
We employ a generic technique which leverages that: (i) each of the processes we analyzed mixes fast and (ii) the effect of many sources can be separated and analyzed separately. 
As a result, we can show that the overall process mixes fast (the  $\bigo(n \log n)$ interactions is as fast as mixing). 
But first, we need to rephrase two technical results from \cite{alistarh2017robust}, adapting them to bi-chromatic setting, where we have two possible initial states.

First, for an agent $u$ of a type $X_i$ or $Y_i$ we denote $\textsf{level}(u) = i$, and if $u$ is in a state $N$ then we define $\textsf{level}(u) = \infty$. We also talk about a \emph{color} of a type of $u$, denoted $\textsf{color}(u)$, being either $X$ or $Y$. If $u$ is of type $N$, we will assign it one of a colors $X$ or $Y$ arbitrarily.
The first is the following Lemma, which adapts a folklore result from load balancing theory to our setting.

\begin{lemma}[\cite{alistarh2017robust}]
\label{lem:resettime}
For any integer $s > 0$, there is constant $C$, such that if there are no agents in state $X$ or $Y$, after $C n \log n$ interactions with high probability $1-n^{-c}$ there is no agent $u$ with $\textsf{level}(u) < s$. The constant $C$ depends only on $s$ and $c$.
\end{lemma}
This Lemma effectively states that system with no $X$ source and no $Y$ source quickly converges to all-$N$ configuration. The following technical tools allow us to use this statement in more complex configurations.

\begin{definition}
Consider two separate populations, $\{u_1,\ldots,u_n\}$ and $\{v_1,\ldots,v_n\}$ each on $n$ agents. We say that the populations are \textbf{coupled}, if in an evolution, after each step $t$,  in each population the
corresponding molecules interact (i.e. the interaction is $u_i + u_j$ and $v_i + v_j$).
\end{definition}

\noindent We now state the following. 
\begin{observation}[c.f. \cite{alistarh2017robust}]
\label{obs:coupling}
Consider $3$ populations $\{u_i\}$, $\{v_i\}$ and $\{w_i\}$ of identical sizes.
For each of the following properties, if it is satisfied after $t$ steps and the corresponding populations are coupled, then it is satisfied indefinitely after each following interaction:
$$\forall_i \textsf{level}(u_i) = \min(\textsf{level}(v_i),\textsf{level}(w_i)).$$  
If initially it holds that $\textsf{level}(u_i) = \textsf{level}(v_i)$ then $\textsf{color}(u_i) = \textsf{color}(v_i)$, and whenever $\textsf{level}(u_i) = \textsf{level}(w_i)$ then $\textsf{color}(u_i) = \textsf{color}(w_i)$, then this property holds indefnitely as well.
\end{observation}

We show the convergence properties in a two-step proof. First, we show that it converges from ``nice'' starting configurations -- where we use coupling argument to compare evolution of freshly started population with already mixed population, for $\bigo(n \log n)$ steps. Second, we bootstrap the argument to all starting populations, by splitting the ``not-nice'' initial conditions away from $X_0$ and $Y_0$ states, using once again coupling argument.

\paragraph{Converging from ``nice'' configurations.} 
Our first technical result will show via a coupling argument that the process converges fast from configurations where all agents are in state $X_0, Y_0$, or $N$.

\begin{lemma}
\label{lem:reset_of_all-n}
Consider a population that is initialized with all agents in states $X_0$, $Y_0$ or $N$, with $X_0$ and $Y_0$ satisfying requirements of Theorem~\ref{th:mainresult}. Then the population will satisfy the guarantees from Theorem~\ref{th:mainresult} after $\bigo(n \log n)$ steps.
\end{lemma}
\begin{proof}
Denote our population as $\{v_i\}$. Denote the number of steps from Lemma~\ref{lem:resettime} as $T_{\textrm{reset}} = \bigo(n \log n)$ and number of steps from Theorem~\ref{th:mainresult} as $T_{\textrm{slow}} = \bigo(n \log^2 n \log \log n)$.

Consider population $\{u_i\}$ initialized identically as population $\{v_i\}$.
We let population $\{u_i\}$ evolve for $T' = T_{\textrm{reset}} + T_{\textrm{slow}}$ steps, starting at step 0. After $T_{\textrm{slow}}$ steps, we take its state (denote it as $S$), and we construct population population $\{w_i\}$, in which for every $u_i \not= X_0,Y_0$ corresponding $w_i$ is set to the identical state, and for $u_i = X_0,Y_0$ we set $w_i$ to be $N$. We set populations $\{v_i\}$ and $\{w_i\}$ to start its evolution at step $T_{\textrm{slow}}$ and make the evolution of populations $\{u_i\}, \{v_i\}$ and $\{w_i\}$ coupled over steps $[T_{\textrm{slow}},T_{\textrm{slow}}+T_{\textrm{reset}}]$. 
Since those populations follow conditions of Observation~\ref{obs:coupling} at 
 $T_{\textrm{slow}}$, the same holds after step $T_{\textrm{slow}}+T_{\textrm{reset}}$. Since by Lemma~\ref{lem:resettime}, after step $T_{\textrm{slow}}+T_{\textrm{reset}}$ $\{w_i\}$ is all $N$ with high probability, conditioned on this high probability event populations $\{u_i\}$ and $\{v_i\}$ are in identical configurations. By  Theorem~\ref{th:mainresult},  since $T' > T_{\textrm{slow}}$, population $\{u_i\}$ reached configuration that satisfies desired bounds. Thus we conclude that population $\{v_i\}$ reached desired bounds at step $T_{\textrm{reset}} = \bigo(n \log n)$.
\end{proof}

\paragraph{Converging from all configurations.} 
Next, we provide the general coupling argument that the process converges fast from configurations where agents are initially in arbitrary state. 

\begin{lemma}
Consider a population that is initialized arbitrarily, with $X_0$ and $Y_0$ satisfying requirements of Theorem~\ref{th:mainresult}. Then the population will satisfy the guarantees from Theorem~\ref{th:mainresult} after $\bigo(n \log n)$ steps.
\end{lemma}
\begin{proof}
Denote our population as $u_i$. Let $\{w_i\}$ be copy of $\{u_i\}$, where each $X_0$ and $Y_0$ is replaced by $N$, and let $\{v_i\}$ be copy of $\{u_i\}$ where each non-$X_0$, non-$Y_0$ is replaced by $N$. Consider coupled evolution of those three populations over next steps $T_{\textrm{reset}}$. By Lemma~\ref{lem:reset_of_all-n}, $\{v_i\}$ reaches configuration that satisfies bounds from Theorem~\ref{th:mainresult}, while by Lemma~\ref{lem:resettime}, $\{w_i\}$ reaches configuration with every agent in state $N$. Applying Observation~\ref{obs:coupling} concludes the proof.
\end{proof}

\section{Extensions}

\subsection{Extension 1: Precision boosting}
\label{sec:precision}
Recall that, in addition to their state, molecules are equipped with a counter that contains an integer value in the interval $[-m, m]$, where $m$ is a parameter. 
The counter is increased by one if a molecule reacts with a strong molecule of type $X_i$, and decreased by one if it reacts with a molecule of type $Y_i$. 
If a molecule reacts with a molecule in state $N$, the counter remains unchanged. 
The output function $\gamma_n$ maps all states with a positive counter to output $X>Y$ and all states with a negative counter to $Y>X$.

In this section, we will sketch the proof of the following result. 
\begin{theorem}
\label{th:expected_stabilization}
If the population satisfies guarantees from Theorem~\ref{th:mainresult}  with $C_3 \ge 2$ for $T = 48 n \ln n$ consecutive steps, then the algorithm reaches at the end of those $T$ steps a configuration with expected number of molecules in wrong output state at most $\bigo(n/\log n)$, with high probability.
\end{theorem}
The proof of this claim will follow from combining two technical sub-claims, Lemma~\ref{lem:boundary_hit} and Lemma~\ref{lem:mixingtime}. The first bounds the rate at which the counter moves towards the correct decision in a stable state. 

\begin{lemma}
\label{lem:boundary_hit}
For any molecule its counter becomes $\log \log n$ at some point (during considered $T$ steps) with high probability.
\end{lemma}
%

The counter behaves like an one dimensional random walk on integers from $- \log \log n$ to $\log \log n$, with bias $b = \frac{\sum_{i=0}^s |X_i|}{\sum_{i=0}^s |Y_i|} \ge 2$ towards +1 steps. 

Let $R$ be a random walk on integers from $\log \log n$ to $-\log \log n$. The transition probability of moving from $i$ to $i+1$ is $p_{i,i+1}=\frac{b}{b+1}$ and the probability of moving from $i$ to $i-1$ is $p_{i, i-1}=\frac{1}{b+1}$. If $i=\log \log n$, $p_{i,i}=\frac{b}{b+1}$ instead of $p_{i,i+1}$. Similarly for $i = -\log \log n$ is $p_{i,i}=\frac{1}{b+1}$ instead of $p_{i,i-1}$.

\begin{lemma}
\label{lem:stationary}
For the stationary distribution $\pi^*$ of the random walk $R$ as defined above, the following holds: 
$$
\frac{\sum_{i=1}^{\log \log n}\pi^*_i}{\sum_{i=-\log\log n}^{-1}\pi^*_i}=
b^{\log\log n+1} \ge 2 \log n.
$$
\end{lemma}

Let $\mathcal{C}$ be a random walk on integers from $\log \log n$ to $(-\log \log n)$ defined such that the state of $\mathcal{C}$ after $t$ steps equals the value of the counter of molecule $m$ after $t$ steps. 
The transition probability of moving from $i$ to $i+1$ at step $t$ is $q_{i,i+1}(t)=\frac{|X_i(t)|}{|X_i(t)|+|Y_i(t)|}$ and the probability of moving from $i$ to $i-1$ is $q_{i, i-1}(t)=\frac{|Y_i(t)|}{|X_i(t)|+|Y_i(t)|}$. If $i=\log \log n$, $p_{i,i}(t)=\frac{|X_i(t)|}{|X_i(t)|+|Y_i(t)|}$ instead of $p_{i,i+1}$. Similarly for $i = -\log \log n$ is $p_{i,i}(t)=\frac{|Y
_i(t)|}{|X_i(t)|+|Y_i(t)|}$ instead of $p_{i,i-1}$.

\begin{lemma}
\label{lem:mixingtime}
If a molecule counter was already equal to $\log \log n$ after $t'$ steps, the probability that the counter is negative at any moment $t \ge t'$ is at most $\frac{1}{\log n}$, for $t,t'$ from $T$ steps from Theorem~\ref{th:expected_stabilization}.
\end{lemma}

We can therefore conclude that in the ``boosted'' version of the comparison protocol, all but an expected $\bigo( n / \log n )$ fraction of the nodes have the correct output. 
We can further boost this result using standard concentration bounds.

\subsection{Extension 2: Time-space trade-offs}
\label{sec:trade-off}

We now provide a sketch of how one can reduce the space cost of the protocols, at the price of slower convergence. First, let us assume in design of protocol we have access to probabilistic transitions. That is, we write
$$A + B \stackrel[1-p]{p}{\rightarrow} \begin{cases} C + D\\ C' + D' \end{cases}$$
to denote that top transition happens with probability $p$ and bottom one with probability $1-p$. For $p<1/2$, this can be simulated by usage of synthetic coin, with roughly $\bigo(\log \log \frac{1}{p})$ extra states.
In our protocol, all probabilistic transitions will be using the same synthetic coin. 

Given such a coin, we will augment our previous protocol with probabilistic interactions such that, fixing probability $p$, roughly, each node at level $i$ will in expectation \emph{inform} roughly $1/p$ nodes before being moved to level $i+1$. 
We then show that there exists a value of the synthetic coin $p'$ such that:
\begin{itemize}
\item The protocol requires $\Theta( \log_{\frac{1}{p'}} n)$ levels,
\item The protocol requires $\Theta(\log \log \frac{1}{p'})$ extra states for the coin,
\item The protocol is slower to converge (w.r.t. to the naive analysis) by a factor of $\frac{1}{p'}$ (this is to be subsumed by analysis in the following subsection)
\item The protocol requires the initial agent count to be larger by a factor of $\bigo(\frac{1}{\sqrt{p'}})$ for the same concentration guarantees.
\end{itemize}


We now provide a sketch of how one can reduce the space usage of the protocols at the cost of slower convergence. First, let us assume in design of protocol we have access to probabilistic transitions. That is, we write
$$A + B \stackrel[1-p]{p}{\rightarrow} \begin{cases} C + D\\ C' + D' \end{cases}$$
to denote that top transition happens with probability $p$ and bottom one with probability $1-p$. For $p<1/2$, this can be simulated by usage of synthetic coin, with roughly $\bigo(\log \log \frac{1}{p})$ extra states.
In our protocol all probabilistic transitions will be using the same synthetic coin. We highlight two standard way of using synthetic coin: 
\begin{enumerate}
\item The coin is stored on extra states, thus the state space is \emph{multiplied} by number of states used by coin. To simulate the coin-flip, agent stores the win/loss bit taken from the last interaction with another agent.
\item  The population is divided into protocol-part and coin-part. Agents flip coin by storing bit win/loss of last interaction with coin-part. The total number of states is roughly the sum number of states of both parts.
\end{enumerate}

\begin{theorem}[\cite{GS18}]
There is a protocol that constructs a synthetic coin, in $\bigo(n \log n)$ interactions. Given parameter $p < 1/2$, it constructs synthetic coin using $\bigo(\log \log \frac{1}{p})$ states and the coin-flip probability is $p'$, where $p^2 \le p' \le p$. It succeeds w.h.p.
\end{theorem}

\subsection{Protocol}
We first provide the protocol for detection.
\begin{equation*}
\begin{aligned}[c]
\text{For all $1\le i \le s:$} \\
U_0 + U_i \rightarrow& U_0 + U_1 \\
\text{~}\\
\text{For all $1\le i < s:$} \\
U_i + N \stackrel[1-p]{p}{\rightarrow}& \begin{cases}U_{i+1} + U_{i+1}\\U_{i} + U_{i+1}\end{cases} \\
\end{aligned}
\qquad
\begin{aligned}[c]
\text{For all $1\le i \le j \le s; i\neq s:$} \\
U_i + U_j \stackrel[1-p]{p}{\rightarrow}& \begin{cases}U_{i+1} + U_{i+1}\\U_{i} + U_{i+1}\end{cases} \\
\text{~}\\
U_s + N \rightarrow& N + N \\
U_s + R_s \rightarrow& N + N \\
\end{aligned}
\end{equation*}
The intuition is that each node at level $i$ is in expectation \emph{informing} roughly $1/p$ nodes before being moved to level $i+1$.

\subsection{Analysis sketch}

Since we follow the pattern set out in the previous sections, we will not present the analysis framework again, and instead focus on the calculations. 
Let us fix a time $t \geq 0$. 
Denote by $\Delta_1(t),\Delta_2(t),\Delta_3(t),\Delta_4(t),\Delta_5(t) \in \{0,1\}$ the indicator variables for the following events at step $t$:
\begin{itemize}
\item $\Delta_1(t) = 1$ iff first reacting agent was from $R_{c+1}$,
\item $\Delta_2(t) = 1$ iff second reacting agent was from $R_{c+1}$,
\item $\Delta_3(t) = 1$ and $\Delta_4(t) = 1$ iff any of the reacting agents were from $R_{c}$.
\item $\Delta_5(t) = 1$ iff any of the reacting agents were from $R_{c+1}$.
\end{itemize}

In the following we use the fact that while $r_c = o(1)$, the square terms do not affect asymptotic of the system and can be discarded. This is justified in this rough analysis, since $r_c = \Theta(1)$ only in a constant number of levels.
Then we obtain the following: 
\begin{align*}
\mathbb{E}\Big[|R_{c+1}(t+1)|\ |\ t\Big] &= p' \cdot \mathbb{E}\Big[|R_{c+1}(t)| - \Delta_1(t) - \Delta_2(t) + \Delta_3(t) + \Delta_4(t)\ |\  t\Big] + \\
&+ (1-p') \cdot \mathbb{E}\Big[|R_{c+1}(t)| - \Delta_1(t) - \Delta_2(t) + \Delta_3(t) + \Delta_5(t)\ |\ t\Big]\\
&= |R_{c+1}(t)| - 2 \cdot r_{c+1}(t) + (1+p') [1 - (1-r_{c}(t))^2]  + (1-p')[1 - (1-r_{c+1}(t))]^2\\
&\approx |R_{c+1}(t)| - 2 p' r_{c+1}(t)  + 2 (1+p')r_{c}(t)\\
&= \left(1-\frac{2p'}{n}\right) |R_{c+1}(t)| + \frac{2p'}{n} \cdot \frac{1+p'}{p'}|R_c(t)|\\
&= \left(1-\frac{2p'}{n}\right) |R_{c+1}(t)| + \frac{2p'}{n} \cdot A(t)
\end{align*}
where we denoted $A(t) = \frac{1+p'}{p'}|R_c(t)|$.
This leads to a steady-state solution of the form $\tilde{R}_{c+1}(t) \approx \frac{1+p'}{p'} \tilde{R}_c(t)$.

For the variance, we obtain
\begin{align*}
\mathrm{Var}\Big[|R_{c+1}(t+1)|\ |\  t\Big] &= \mathrm{Var}\Big[|R_{c+1}(t+1)| - |R_{c+1}(t)|\ |\ t\Big] \\
&\le 4\Big(\mathrm{Var}\Big[\Delta_1(t) | t\Big] + \mathrm{Var}\Big[\Delta_2(t) | t\Big] + \mathrm{Var}\Big[\Delta_3(t) | t\Big] + p' \mathrm{Var}\Big[\Delta_4(t) | t\Big] + (1-p')\mathrm{Var}\Big[\Delta_5(t) | t\Big]\Big) \\
&\le 4\Big(\E\Big[\Delta_1(t)^2 | t\Big] + \E\Big[\Delta_2(t)^2 | t\Big] + \E\Big[\Delta_3(t)^2 | t\Big] + p'\E\Big[\Delta_4(t)^2 | t\Big] + (1-p')E\Big[\Delta_5(t)^2 | t\Big]\Big)\\
&= 4\Big(\E\Big[\Delta_1(t) | t\Big] + \E\Big[\Delta_2(t) | t\Big] + \E\Big[\Delta_3(t) | t\Big] + p'\E\Big[\Delta_4(t) | t\Big] + (1-p')\E\Big[\Delta_5(t) | t\Big]\Big)\\
&\approx 4 \Big(r_{c+1}(t) + r_{c+1}(t) + 2(1+p')r_c(t) + 2(1-p')r_{c+1}(t)\Big)\\
&\le 16 \frac{|R_{c+1}(t)|}{n} + 8p' \frac{A(t)}{n}.
\end{align*}

We can thus invoke Theorem~\ref{th:general_martingale} with $a = \tilde{R}_{c+1}$, $\lambda = 2p'$, $\delta = 16$, $\eta = 8p'$ and $\gamma=1$. Thus we get that if there is on level $c$ an absolute error $\tilde{R}_{c} \cdot \varepsilon$, on level $c+1$ the error becomes at most $\tilde{R}_{c+1} \cdot \left( \varepsilon + \bigo\left(\sqrt{\frac{1}{p'} \frac{\log n}{\tilde{R}_{c+1}}}\right)\right)$. The time for stabilizing single level becomes then $\Theta( \frac{1}{p'} n \log n \log \log n)$. Thus we draw a conclusion that in total time $\Theta( \frac{1}{p'} n \log^2 n \log \log n)$ we reach total error $\bigo(n \sqrt{\frac{1}{p'} \frac{\log n}{\tilde{R}_0}})$. We summarize the take-away message of this subsection:
Protocol supplemented with synthetic coin $p'$ is:
\begin{itemize}
\item requires $\sim \log_{\frac{1}{p'}} n$ levels,
\item requires $\Theta(\log \log \frac{1}{p'})$ extra states for synthetic coin,
\item slower to converge (wrt to naive analysis) by a factor of $\frac{1}{p'}$ (this is to be subsumed by analysis in the following subsection)
\item requires initial agent count to be larger by a factor of $\bigo(\frac{1}{\sqrt{p'}})$ for the same concentration guarantees.
\end{itemize}
\subsection{Bootstraping}
We now comment on how bootstrapping analysis translates to this new setting. First, we extend definition of coupling so that random choices of all coupled protocols are identical. We then analyze the decay time for the new protocol.
\begin{lemma}
For any integer $s>0$ and $c\ge1$ if there are no agents $u$ with $level(u) = 0$ then after $\Theta(\frac{1}{p'} c n \log n + \frac{1}{p'^2}s n)$ interactions with high probability $1- n^{-c}$ there is no agent $u$ with $level(u) < s$. \end{lemma}
\begin{proof}
For agent $u$ we define its potential (at time $t$) to be $\Phi_t(u) = d^{- level_t(u)}$, where $d = 2\frac{1+p'}{p'}$. We define potential of a whole population to be $\Phi_t = \sum_{u} \Phi_t(u)$. Consider two agents $u,v$ at arbitrary time $t$, with $level_t(u) = x$ and $level_t(v) \ge x$ for some constant $x$. We have then,
$$\mathbb{E}[ \Phi_{t+1}(u) + \Phi_{t+1}(v)\ |\ t ] = p' \cdot 2 d^{-(x+1)} + (1-p') \cdot (d^{-x} + d^{-(x+1)}) = d^{-x}\left(\frac{1+p'}{d} + (1-p')\right) \le \alpha \left(\Phi_t(u)+\Phi_t(v)\right).$$
where $\alpha = \left(\frac{1+p'}{d} + (1-p')\right) = (1-\frac{p'}{2})$.
$$\mathbb{E}[ \Phi_{t+1}\ |\ t] \le \sum_{u} (1-\frac{2}{n}) \Phi_t(u) + \frac{2}{n} \alpha \Phi_t(u) = \left(1 - \frac{2}{n}(1-\alpha)\right) \Phi_t = \left(1 - \frac{p'}{n}\right) \Phi_t,$$
so by submartingale property $\mathbb{E}[ \Phi_t ] \le  \left(1 - \frac{p'}{n}\right)^t n.$
Denote by $T =  \left((c+1) \ln(n) + s \ln(d)\right) \frac{n}{p'} = \Theta(\frac{1}{p'} c n \log n + \frac{1}{p'^2}s n)$. We have then $\mathbb{E}[\Phi_T] \le n^{-c} \cdot d^{-s}$, so $\Pr[\Phi_T \ge d^{-s}] \le n^{-c}$.

\end{proof}

Substituting $s = \bigo(\log n)$ and $c = 3$, and repeating coupling analysis, we reach that a concentration happens in time $T = \bigo(\frac{1}{p'^2}n \log n)$.
We can now examine what this result yields for some non-trivial parameter regimes:
\begin{itemize}
\item $p' = \frac{1}{2^{\sqrt{\log \log n}}}$, with $\Theta(\log \log \log n)$ states for the coin, $\frac{\log n}{\sqrt{\log \log n}}$ levels, and convergence in $ \log^{1+o(1)} n$  parallel time.
\item $p' = \frac{1}{\log^{\Theta(1)} n}$, with $\Theta(\log \log \log n)$ states for the coin, $\frac{\log n}{\log \log n}$ levels, and convergence in $\log^{\Theta(1)} n$ parallel time.
\item $p' = \frac{1}{2^{\sqrt{\log n}}}$, with $\Theta(\log \log n)$ states for the coin, $\sqrt{\log n}$ levels, and convergence in $2^{\sqrt{\log n}}$ parallel time.
\item $p' = \frac{1}{2^{(\log n)/(\log \log n)}}$ with $\Theta(\log \log n)$ states for the coin, $\log \log n$ levels, and convergence in $n^{o(1)}$ parallel time.
\end{itemize}

\subsection{Extension 3: Leaks}
We now consider a scenario where leaks can occur. That is, there are occurring (possibly adversarially) spontaneous reactions of type $A \to B$ for some non-catalytic states $A,B$. The states $X_0$ and $Y_0$ are \emph{catalytic} states in our setting--as they are not created or modified by the algorithm, only detected--they are not affected by leaks. 
That is, $X_0$ and $Y_0$ cannot spuriously appear or disappear. 
However, all other states may be affected by leaks, and therefore may appear or disappear spuriously, as a consequence of leaks. 
More precisely, the protocol is subject to arbitrary reactions of the type $A \to B$, where $A$ and $B$ are arbitrary nodes in states outside $\{X_0, Y_0\}$. 
At the same time, as in~\cite{alistarh2017robust}, we assume that the rate of leaks is bounded by a parameter $\zeta$. 
Upon reflection, we notice that, in the comparison problem, the strongest adversarial leak strategy would be to leak first-level strong states from the majority state to the minority one: say $X_1 \to Y_1$.

\begin{corollary}
\label{cor:leaks}
If the leak probability $\zeta$ is $\bigo(1/n)$, then the asymptotic guarantees for detection and comparison are asymptotically identical to the ones provided in the leakless case.
\end{corollary}

\subsection{False-positive leaks}
We refine the analysis of the detection protocol, noting that incorporating false-positive leaks into the protocol, it takes following form.

\begin{equation*}
\begin{aligned}[c]
\text{For all $1\le i \le s:$} \\
U_0 + U_i \rightarrow& U_0 + U_1 \\
\text{~}\\
\text{For all $1\le i < s:$} \\
U_i + N \stackrel[\zeta]{1-\zeta}{\rightarrow}& \begin{cases}U_{i+1} + U_{i+1}\\U_{1} + U_{1}\end{cases} \\
\end{aligned}
\qquad
\begin{aligned}[c]
\text{For all $1\le i \le j \le s; i\neq s:$} \\
U_i + U_j \stackrel[\zeta]{1-\zeta}{\rightarrow}& \begin{cases}U_{i+1} + U_{i+1}\\U_{1} + U_{1}\end{cases} \\
U_s + N \stackrel[\zeta]{1-\zeta}{\rightarrow}& \begin{cases}N + N\\U_{1} + U_{1}\end{cases} \\
U_s + U_s  \stackrel[\zeta]{1-\zeta}{\rightarrow}& \begin{cases}N + N\\U_{1} + U_{1}\end{cases}\\
\end{aligned}
\end{equation*}

In particular, recall the notation where we denote by $\Delta_1(t),\Delta_2(t),\Delta_3(t),\Delta_4(t) \in \{0,1\}$ the indicator variables for the following events at step $t$, which govern the evolution of $|R_{c + 1}(t + 1)|$:
\begin{itemize}
\item $\Delta_1(t) = 1$ iff first reacting agent was from $R_{c+1}$,
\item $\Delta_2(t) = 1$ iff second reacting agent was from $R_{c+1}$,
\item $\Delta_3(t) = 1$ and $\Delta_4(t) = 1$ iff any of the reacting agents were from $R_{c}$.
\end{itemize}

We obtain the following recurrence on the expected value of $|R_{c + 1}(t + 1)|$: 
\begin{align*}
\mathbb{E}\Big[|R_{c+1}(t+1)|\ |\ t\Big] &= \mathbb{E}\Big[|R_{c+1}(t)| - \Delta_1(t) - \Delta_2(t)\ |\  t\Big] + (1-\zeta)\mathbb{E}[\Delta_3(t) + \Delta_4(t)\ |\ t] + 2\zeta\\
&= |R_{c+1}(t)| - 2 r_{c+1}(t) + 2(1-\zeta) [1 - (1-r_{c}(t))^2] + 2 \zeta \\
&= (1-\frac{2}{n}) |R_{c+1}(t)| + \frac{2}{n} A'(t)
\end{align*}
where we define $A'(t) = n \cdot [1 - (1-\zeta)(1-r_{c}(t))^2)]$. This leads to a steady-state solution of exact form
$1-\tilde{r}'_c = (1-\zeta)^{2^c-1} (1-\tilde{r}_0)^{2^c}.$
\subsection{False-negative leaks}
We refine the analysis of the detection protocol, noting that incorporating false-positive leaks into the protocol, it takes following form.

\begin{equation*}
\begin{aligned}[c]
\text{For all $1\le i \le s:$} \\
U_0 + U_i  \stackrel[\zeta]{1-\zeta}{\rightarrow}& \begin{cases}U_{0} + U_{1}\\U_0+N\end{cases} \\
\text{For all $1\le i < s:$} \\
U_i + N \stackrel[\zeta]{1-\zeta}{\rightarrow}& \begin{cases}U_{i+1} + U_{i+1}\\N+N\end{cases} \\
\end{aligned}
\qquad
\begin{aligned}[c]
\text{For all $1\le i \le j \le s; i\neq s:$} \\
U_i + U_j \stackrel[\zeta]{1-\zeta}{\rightarrow}& \begin{cases}U_{i+1} + U_{i+1}\\ N+N\end{cases} \\
\text{~}\\
U_s + N &\rightarrow N+N \\
U_s + U_s  &\rightarrow N+N\\
\end{aligned}
\end{equation*}

We obtain the following recurrence on the expected value of $|R_{c + 1}(t + 1)|$: 
\begin{align*}
\mathbb{E}\Big[|R_{c+1}(t+1)|\ |\ t\Big] &= \mathbb{E}\Big[|R_{c+1}(t)| - \Delta_1(t) - \Delta_2(t)\ |\  t\Big] + (1-\zeta)\mathbb{E}\Big[\Delta_3(t) + \Delta_4(t)\ |\ t\Big] + 2\zeta\\
&= |R_{c+1}(t)| - 2 r_{c+1}(t) + 2(1-\zeta) [1 - (1-r_{c}(t))^2] \\
&= (1-\frac{2}{n}) |R_{c+1}(t)| + \frac{2}{n} A''(t)
\end{align*}
where we define $A''(t) = n \cdot(1-\zeta) [1 - (1-r_{c}(t))^2)]$. This leads to a steady-state solution 
$\tilde{r}''_{c+1} \approx (1-\zeta)2 \tilde{r}''_c$
\subsection{Concentration}
Second, we bound the variance by direct calculation, which applies to both types of leaks.
\begin{align*}
\mathrm{Var}\Big[|R_{c+1}(t+1)|\ |\  t\Big] &= \mathrm{Var}\Big[|R_{c+1}(t+1)| - |R_{c+1}(t)|\ |\ t\Big] \\
&\le 4\Big(\mathrm{Var}\Big[\Delta_1(t) | t\Big] + \mathrm{Var}\Big[\Delta_2(t) | t\Big] + (1-\zeta)\mathrm{Var}\Big[\Delta_3(t) | t\Big] + (1-\zeta)\mathrm{Var}\Big[\Delta_4(t) | t\Big]\Big)\\
&\le 4 (2r_{c+1}(t)+ 2(1-\zeta)(1-(1-r_c(t))^2)) \\
&\le 8 \frac{|R_{c+1}(t)|}{n} + 8 \frac{A''(t)}{n}\\
&\le 8 \frac{|R_{c+1}(t)|}{n} + 8 \frac{A'(t)}{n}.
\end{align*}
An application of Theorem~\ref{th:general_martingale} follows, giving the same concentration in both cases around $\tilde{r}'_c$ and $\tilde{r}''_c$ respectively. Those two cases actually represent upper- and lower- bounds on possible steady-state solutions. We have
$\tilde{r}'_c \approx \zeta\cdot 2^c + \tilde{r}_0 \cdot 2^c$ and $\tilde{r}''_c \approx (1-\zeta)^c 2^c \tilde{r}_0$,
thus the additional spread introduced is $\approx \zeta \cdot 2^c + c \zeta 2^c \tilde{r}_0 = \bigo(\zeta n)$.

\section{Conclusion and Future Work}

We have introduced the comparison problem, and presented a simple dynamics to solve this problem in a self-stabilizing and robust manner. 
A valid alternative view of our results is \emph{algorithm-oriented}, rather than \emph{problem-oriented}: we are studying of a natural comparison dynamics in the protocol model. 

The algorithm we consider is guaranteed to converge to configurations where at most $\bigo( n / \log n)$ of the agents are in the ``wrong'' state. Hence, a single random sample from the stable solution will return the correct output with probability $\geq 1 - \bigo(1 / \log n)$. Notice that this is in some sense the best we could hope for with this type of dynamics, since $\Theta(n / \log n)$ agents will be in the strong state for the minority opinion by the structure of the chain.
Further, our protocol uses $\bigo( \log n \log \log n)$ states, and converges in $\bigo(\log n)$ parallel time, w.h.p. We believe that the state space can be improved to $\bigo( \log n),$ which we leave for future work.

As future work, it would be interesting to consider whether further simplified dynamics exist, which would remove the need for the additional counter state used to decide the output opinions. It would also be interesting to extend this binary comparison problem to multiple baseline states. While our analysis will immediately extend to the case where the count of the most populous baseline state dominates the sum of all the other baseline states, it would be interesting to examine whether finer-grained convergence conditions exist, i.e. the case of plurality. Finally, 
another interesting question for future work regards improved lower bounds for the comparison problem. 

\section*{Acknowledgments}

We would like to thank Rati Gelashvili for very useful discussions, and the PODC anonymous reviewers for their careful reading of our paper, and for their useful remarks.
This work is partially supported by the Polish National Science Center (NCN) grant
UMO-2017/25/B/ST6/02010.

\bibliographystyle{alpha}
\bibliography{references}

\newpage

\appendix

\section{Proof of Theorem~\ref{th:general_martingale}.}

\textsc{Theorem~\ref{th:general_martingale}.}\emph{\maintoolA\maintoolB\maintoolC}

\vspace{1em}
We start with the following lemma
\begin{lemma}
\label{onephaselemma}
Let $T = \Theta(\frac{1}{\lambda} n \log n)$ and additionally assume $B(t) \le m$ for $t \in [T]$, for $m \ge a$. Then $|B(T) - a|  \le \bigo(c_1 \sqrt{m \log n} + \log n) + \varepsilon n$ with high probability, where $c_1 = \sqrt{\frac{\delta+\eta}{\lambda}} $.
\end{lemma}
\begin{proof}

Denote 
$$\Delta(t+1) = B(t+1) - (1-\frac{\lambda}{n}) B(t) - \frac{\lambda}{n} A(t).$$
 There is
$$\E[ \Delta(t+1)\ |\ A(t),B(t) ] =0,$$
$$\Var[ \Delta(t+1)\ |\ A(t),B(t) ] \le \delta \cdot \frac{B(t)}{n} + \eta \cdot \frac{A(t)}{n}$$
and
$$|\Delta(t+1)| \le  (1-\frac{\lambda}{n}) |B(t+1)-B(t)| + \frac{\lambda}{n} |B(t+1)-A(t)| \le \gamma + \lambda.$$

Denote $$A =  \sum_{j=0}^{T-1} \frac{\lambda}{n} (1-\frac{\lambda}{n})^{j} A(T-j-1).$$

We proceed to bound the random variable
\begin{align*}
\Phi &= \sum_{j=0}^{T-1} (1-\frac{\lambda}{n})^j \cdot  \Delta(T-j)\\
&= \sum_{j=0}^{T-1} (1-\frac{\lambda}{n})^j \cdot B(T-j) - \sum_{j=0}^{T-1} (1-\frac{\lambda}{n})^{j+1} \cdot B(T-j-1)  - \sum_{j=0}^{T-1} \frac{\lambda}{n} (1-\frac{\lambda}{n})^{j} A(T-j-1)\\
 &= B(T) - (1-\frac{\lambda}{n})^T B(0) - A
\end{align*}

We first observe:
$$
\E[\Phi] = \sum_{j=0}^{T-1} \E[ \Delta(T-j) \cdot (1-\frac{\lambda}{n})^j\ |\ A(T-j-1), B(T-j-1) ] = 0
$$
And we bound the sum of conditional variances:
\begin{align*}
K &= \sum_{j=0}^{T-1} \Var[ \Delta(T-j) \cdot (1-\frac{\lambda}{n})^j\ |\ A(T-j-1),B(T-j-1) ]\\
&\le \sum_{j=0}^{T-1} \left(\delta \cdot \frac{B(T-j)}{n} + \eta \cdot \frac{A(T-j)}{n}\right)\cdot (1-\frac{\lambda}{n})^{2j}\\
&\le \left(\delta \cdot \frac{m}{n} + \eta \frac{(1+\varepsilon) a}{n} \right) \cdot \frac{n}{\lambda}\\
&\le \frac{\delta + 2\eta}{\lambda} \cdot m
\end{align*}

We also state the absolute variables bound:
$$ M = \max_{j=0,..,T-1} | \Delta(T-j) \cdot (1-\frac{\lambda}{n})^j| \le \gamma + \lambda.$$

By Bernstein's inequality for martingales~\cite{Bernstein}
$$\Pr[|\Phi| \ge t] \le 2 \exp\left(\frac{-t^2/2}{K + M t/3}\right) \le \exp( \frac{-t^2}{2K} ) + \exp( \frac{-3t}{2M} )$$

It is thus enough to set $t = \Theta(\sqrt{K \log n} + M \log n) = \Theta(c_1\sqrt{m \log n} + \log n)$ for the bound to have $|\Phi| \le t$ with high probability.

We then observe that (by using appropriate bound on $A(T-j)$ and sums of geometric progressions)
$$A  \le (1+\varepsilon )\cdot a$$
$$A \ge (1- (1-\frac{\lambda}{n})^T)(1-\varepsilon) \cdot a \ge (1-n^{-\Theta(1)})(1-\varepsilon)  a$$
thus
$$| a - A | \le \varepsilon a + n^{-\Theta(1)} $$

Thus following holds
\begin{align*}
|B(T) - a| &\le |\Phi| + |a - A| + (1-\frac{\lambda}{n})^T B(0)\\
&= \bigo(c_1 \sqrt{m \log n} + \log n) + \varepsilon a.
\end{align*}
 \end{proof}
 
We now iterate $\Theta(\log \log n)$ times Lemma \ref{onephaselemma}  to bootstrap the concentration. 
We proceed in phases, where each phase is of length $T$ required for Lemma~\ref{onephaselemma} to work, and phase $k$ spans $T_k = [T\cdot (k-1)+1,T \cdot k]$. Let $m_k = \max_{t \in T_k} B(t)$. Initially we trivially have $m_0 \le n$.

We observe that Lemma~\ref{onephaselemma} applied to phase $k$ reduces upperbound of $m_k$ to $m_{k+1} \le (1 + \varepsilon )a + \bigo (c_1 \sqrt{m_k \log n} + \log n) \le C \max( c_1 \sqrt{m_k \log n}, a)$ for some  constant $C$.

By easy inductive argument it follows that
$$m_{k} \le \max\left( \left((c_1C)^2 \log n\right)^{1-2^{-k}} \cdot (m_0)^{2^{-k}}, C a\right).$$
For some $\ell = \Theta(\log \log n)$ there is (since $a = \Omega(\log n)$)
$$m_{\ell} = \bigo(C a + (c_1C)^2 \log n) =  \bigo( a + c_1 \log n).$$
which gives us that for time $t' \ge T' = T \cdot \ell$, by Lemma~\ref{onephaselemma}
$$|B(t') - a | \le \varepsilon a + \bigo(c_1 \sqrt{a \log n} + \log n).$$
\qed

\end{document}